\documentclass[12pt]{article}

\usepackage[latin1]{inputenc}
\usepackage{amssymb} 
\usepackage{lscape} 
\usepackage[colorlinks,bookmarks=true]{hyperref}
\usepackage{caption}
\usepackage{amsmath}
\usepackage{amsthm}
\usepackage{latexsym}
\usepackage{subfigure}
\usepackage{graphicx}
\usepackage{bm}
\usepackage{bbm}
\usepackage{overpic}
\usepackage[normalem]{ulem}
\usepackage{url}
\usepackage{soul}
\usepackage{exscale}
\usepackage{amsfonts}
\usepackage[usenames,dvipsnames]{color} 
\usepackage{verbatim}
\usepackage[usenames,dvipsnames]{xcolor}
\usepackage{comment}
\usepackage{epsfig}

\newtheorem{lemma}{Lemma}[section]
\newtheorem{remark}{Remark}[section]

\newtheorem{theorem}{Theorem}[section]

\newtheorem{prop}{Proposition}[section]

\numberwithin{equation}{section}
\numberwithin{figure}{section}
\numberwithin{table}{section}

\newcommand{\E}{\mathbb{E}}
\newcommand{\R}{\mathbb{R}}

\newcommand{\beas}{\begin{eqnarray*}}
\newcommand{\eeas}{\end{eqnarray*}}
\newcommand{\bea}{\begin{eqnarray}}
\newcommand{\eea}{\end{eqnarray}}
\newcommand{\ben}{\begin{enumerate}}
\newcommand{\een}{\end{enumerate}}

\newcommand{\cF}{\mathcal{F}}

\newcommand{\cC}{\mathcal{C}}

\def \N{\mathbb{N}}

\newcommand{\mQ}{\mathbb{Q}}

\newcommand{\mP}{\mathbb{P}}

\newcommand{\bi}{\begin{itemize}}
\newcommand{\ei}{\end{itemize}}
\newcommand{\beq}{\begin{equation}}
\newcommand{\eeq}{\end{equation}}
\newcommand{\bv}{\begin{verbatim}}
\newcommand{\ev}{\end{verbatim}}

\newcommand{\ee}[1]{\ensuremath{\mathbb{E}\left[{#1}\right]}}
\newcommand{\eef}[1]{\ensuremath{\mathbb{E}\left[\left.{#1}\right|\cF_t\right]}}

\newcommand{\angl}[1]{\ensuremath{\langle{#1}\rangle}}

\title{The Zumbach effect under rough Heston}

\author{Omar El Euch, \'Ecole Polytechnique,\\ {\tt  omar.el-euch@polytechnique.edu}\\~\\
Jim Gatheral, Baruch College, CUNY,\\ {\tt jim.gatheral@baruch.cuny.edu}\\~\\
Rado\v{s} Radoi\v{c}i\'c, Baruch College, CUNY,\\ {\tt rados.radoicic@baruch.cuny.edu}\\~\\
Mathieu Rosenbaum,  \'Ecole Polytechnique,\\ {\tt mathieu.rosenbaum@polytechnique.edu}
 }

\begin{document}

\maketitle

\begin{abstract}   

Previous literature has identified an effect, dubbed the Zumbach effect, that is nonzero empirically but conjectured to be zero in any conventional stochastic volatility model. Essentially this effect corresponds
to the property that past squared returns forecast future volatilities better than past volatilities forecast future squared returns. We provide explicit computations of the Zumbach effect under rough Heston and show that they are consistent with empirical estimates. In agreement with previous conjectures however, the Zumbach effect is found to be negligible in the classical Heston model.%
\end{abstract}


\section{Introduction}\label{intro}

In a series of papers \cite{borland2005dynamics, lynch2003market, zumbach2001heterogeneous, zumbach2004volatility, zumbach2009time}, Gilles Zumbach and co-authors identified several empirical features of financial time series that are not well replicated by conventional stochastic volatility models.  In this paper, we focus on one particular such effect dubbed the {\em Zumbach effect} in \cite{blanc2017quadratic}.


Denote the true integrated variance from the open to the close of day $t$ by $\sigma^2_t$, the open to close return by $r_t$, and let $\angl{\cdot}$ represent a sample average.  Then, for $\tau \in \R$, the statistic (6b) of \cite{chicheportiche2014fine}
\[
\tilde \cC^{(2)}(\tau) = \angl{\left({\sigma_t}^2- \angl{{\sigma_t}^2}\right)\,r^2_{t-\tau}    }  
\]
quantifies (under stationarity assumptions) the covariance of integrated variance with past squared returns.  The particular measure of time-reversal asymmetry (TRA) that is found empirically to be positive in \cite{chicheportiche2014fine} is given by
\beq
\label{eq:ztau}
Z(\tau) := \tilde \cC^{(2)}(\tau) - \tilde \cC^{(2)}(-\tau), \quad \tau > 0.
\eeq
In words, the covariance between historical squared returns and future integrated variance is greater than the covariance between historical integrated variance and future squared returns.  The following quote from  \cite{blanc2017quadratic} refers to this measure $Z(\tau)$ of TRA:
\begin{quote}
Interestingly, all continuous time stochastic volatility models, from the famous CIR-Heston model (Cox et al. 1985, Heston 1993) to the Multifractal Random Walk model alluded to above, obey TRS\footnote{Time-reversal symmetry} by construction and therefore {\it cannot} account for the empirical TRA of financial time series. 
\end{quote}

In the present paper, we first confirm that $Z(\tau)$ is empirically nonzero.  We then compute $Z(\tau)$ explicitly under rough Heston.
We show that when the Hurst parameter $H$ of the volatility process is small ($H$ of order $0.1$) as established empirically in \cite{gatheral2018volatility} and confirmed in \cite{bennedsen2016decoupling}, the Zumbach effect obtained under rough Heston is very consistent with empirical estimates.
However, when $H=1/2$, corresponding to the conventional Heston model, we get that $Z(\tau)$ is indeed numerically absolutely negligible.

Our paper proceeds as follows.  In Section \ref{sec:empirical}, we confirm that the Zumbach effect is empirically nonzero. In Section \ref{sec:roughHeston}, we compute the Zumbach effect under rough Heston.  Finally, in Section \ref{sec:numerical}, we show that the rough Heston model is both qualitatively and quantitatively  consistent with empirical estimates. Some additional detailed computations are relegated to the appendix.

\section{Empirical estimation of the Zumbach effect}\label{sec:empirical}

For our empirical study, we use opening and closing prices and precomputed realized kernel estimates of intraday (open to close) integrated variance from the Oxford-Man Institute of Quantitative Finance Realized Library  from 2000, January 3 to 2018, July 25.\footnote{\url{http://realized.oxford-man.ox.ac.uk/data/download} The Oxford-Man Institute of Quantitative Finance Realized Library contains a selection of daily non-parametric estimates of volatility of financial assets, including realized variance and realized kernel estimates.  A selection of such estimators is described and their performances compared in, for example, \cite{gatheral2010zero}.}.

There are 31 indices in the Oxford dataset, as listed in Appendix \ref{sec:indices}.  We proceed by computing $\tilde C^{(2)}(\tau)$ and $\tilde C^{(2)}(-\tau)$ for each of these indices and converting these to correlations by dividing by the relevant sample variances.  That is, for each index, we compute
\[
\tilde \rho(\tau) = \frac{\tilde C^{(2)}(\tau)}{\sqrt{\angl{(\sigma_t^2-\angl{\sigma_t^2})^2}\,\angl{(r_{t-\tau}^2-\angl{r_{t-\tau}^2})^2}}}.
\]
We then average the $\tilde \rho_j$ across the indices $j$ in the dataset to obtain 
\[
\bar \rho (\tau) = \frac1 {31}\sum_{j=1}^{31}\,\tilde \rho_j(\tau).
\]
Finally, corresponding to Equation (25) of \cite{chicheportiche2014fine}, we further define the integrated difference
\[
\Delta(\tau) = \sum_{i=1}^\tau\,\left(  \bar \rho(i) - \bar \rho(-i)  \right).
\]
In Figure \ref{fig:chicheFig10}, we present respectively $\bar \rho(\tau)$, $\bar \rho(-\tau)$ and $\Delta(\tau)$, reproducing Figure 10 of \cite{chicheportiche2014fine}, and confirming empirically that the Zumbach effect is nonzero.
 \begin{figure} [tbh!]
\centering
\includegraphics[width=\linewidth ]{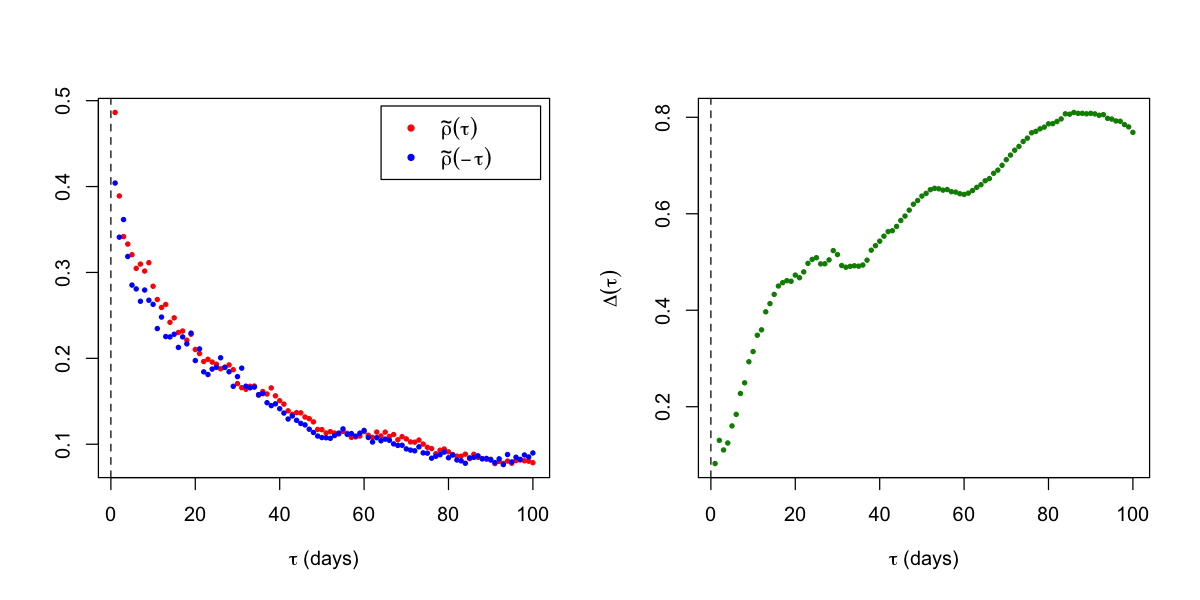}
\caption{In the left panel, the $\bar \rho(\tau)$ are in red, the $\bar \rho(-\tau)$ in blue. In the right panel we plot the integrated difference $\Delta(\tau)$.}
\label{fig:chicheFig10}
\end{figure}

\begin{remark}
Comparing Figure \ref{fig:chicheFig10} with Figure 10 of \cite{chicheportiche2014fine}, we note that our correlations are in general significantly greater.  We attribute this difference to the superior accuracy of the Oxford-Man realized kernel estimates of integrated variance that we use here relative to the Rogers-Satchell estimates computed in \cite{chicheportiche2014fine}.
\end{remark}



\section{The rough Heston model}\label{sec:roughHeston}

\subsection{Description of the model}

We consider the rough Heston model introduced in \cite{el2018characteristic} for the price $S_t$ of an asset at time $t$:
$$ dS_t = S_t \sqrt{V_t} dW_t, $$
\begin{equation}\label{eq:variance} V_t =  g_0(t) -  \frac{1}{\Gamma(H + 1/2)}  \int_{0}^t (t-s)^{H-1/2} \lambda V_s ds +\frac{1}{\Gamma(H + 1/2)}  \int_{0}^t (t-s)^{H-1/2} \nu \sqrt{V_s} dB_s.
\end{equation}
Here $H \in (0, 1/2]$ is the Hurst exponent of the volatility, $\lambda>0$ is the mean reversion parameter, $\nu>0$ is the volatility of volatility parameter and $(W, B)$ is a $\rho$-correlated Brownian motion with $\rho \in [-1, 1]$. The function $g_0$ is assumed to be continuous and is linked to the forward variance curve $\xi_0(t) = \E[V_t]$ as follows:
\beq 
g_0(t) = \xi_0(t) +  \frac{1}{\Gamma(H + 1/2)}  \int_{0}^t (t-s)^{H-1/2}\, \lambda \,\xi_0(s) ds . \label{eq:g0}
\eeq
Note that in \cite{jaber2018markovian}, a general condition on $g_0$ is given to guarantee weak existence and uniqueness for the solution of the equation defining the rough Heston model. 

 In \cite{el2018characteristic,euch2018perfect}, it is shown that there exists a semi-closed form expression for the characteristic function, just as in the classical Heston case, and that explicit hedging strategies can be derived. Fast and accurate option pricing is also possible, see \cite{gatheral2018rational}.  
Furthermore, the rough Heston model displays the rough behavior of the volatility observed empirically in \cite{gatheral2018volatility}.  
  More precisely, the variance process $V$ admits H\"older continuous paths with regularity strictly less than $H$. In addition to the fit to historical data, it is shown in \cite{el2018roughening}  that with suitably calibrated parameters ($H$, $\nu$, $\rho$ and $\lambda$), the rough Heston model typically fits the SPX volatility surface remarkably well.

\subsection{The Zumbach effect under rough Heston: explicit computation}

We provide in this section an explicit formula for the Zumbach effect in the rough Heston model (Theorem \ref{theo_approx}). We start with a discussion about the use of correlations or covariances when computing the Zumbach effect under rough Heston.

\subsubsection{Correlations versus covariances}

From a theoretical viewpoint, approximating theoretical quantities such as covariances and correlations by sample values makes sense only provided the underlying dynamics can be considered stationary. In the context of the rough Heston model \eqref{eq:variance}, this would imply that the parameter $\lambda$ should be large enough with respect to the observation time scale. However, whether rough volatility models are estimated under $\mP$ or $\mQ$, $\lambda$ is typically found to be  small relative to this observation time scale, 
see \cite{el2018roughening}. In this case, under rough Heston, the very notion of the Zumbach effect may appear somehow ill-defined.

%

It turns out, as will be seen in Section \ref{computation_Zumbach}, that under rough Heston, the Zumbach effect $Z(\tau)$  expressed as a difference of covariances \eqref{eq:ztau}, does not depend asymptotically on $\lambda$. This is in contrast to the effect $Z^{\text{Correl}}(\cdot) $ expressed in terms of corrrelations as in Proposition \ref{prop:zumbach_prop}.
%
 Consequently, we choose to focus on covariances and express the Zumbach effect in terms of the covariances $\tilde C^{(2)}$ rather than the correlations $\tilde \rho$. For the sake of completeness, computations based on correlation in the stationary regime are presented in Appendix \ref{sec:appcorrel}. 

\subsubsection{Computation of the Zumbach effect}\label{computation_Zumbach}

Denote the length of the trading day by $\delta$.  Under the rough Heston model, the open to close (log-)return is given by\footnote{For simplicity, we take only the martingale part of the log-returns.}
$$ r_t = \int_{t - \delta}^t \sqrt{V_s} dW_s, $$
and the daily integrated variance by 
$$ \sigma^2_t = \int_{t- \delta}^t V_s ds. $$

\noindent Our aim is to prove that in the rough Heston model, the counterpart of $Z(\tau) = \tilde \cC^{(2)}(\tau) - \tilde \cC^{(2)}(-\tau)$ is positive
for $\tau = k \,\delta$ with $k \in \N_{>0}$. Hence, we write
$$ Z_t(k) =  \text{Cov}[r_{t}^2, \sigma^2_{t+k \delta}] - \text{Cov}[r_{t + k \delta}^2, \sigma^2_{t}], \quad k \in \N_{>0}, \quad t \geq \delta.$$

\noindent  By It\^o's isometry, $ \ee{r_t^2}=\ee{\sigma^2_t}$ for any time $t$ so
$$ 
Z_t(k) = \E[r_t^2 \,\sigma^2_{t + k \delta}] - \E[r_{t + k \delta}^2\, \sigma^2_t]. 
$$
Applying It\^o's formula, we get 
$$ r_t^2 = 2 \int_{t- \delta}^t \sqrt{V_s} \left( \int_{t- \delta}^s \sqrt{V_u} dW_u \right) dW_s  + \sigma^2_t . $$
This together with the fact that $\ee{r_{t + k\delta}^2 \,\sigma^2_{t}} =  \ee{\eef{r_{t + k \delta}^2}\, \sigma^2_{t}} =  \E[\sigma^2_{t + k \delta} \,\sigma^2_t] $ leads to
\begin{equation}\label{inter}
Z_t(k) = 2\, \E\left[ \sigma^2_{t + k \delta} \int_{t- \delta}^t \sqrt{V_s} \left( \int_{t- \delta}^s \sqrt{V_u} dW_u \right) dW_s \right]. 
\end{equation}
\noindent Substituting \eqref{eq:g0} into \eqref{eq:variance} with $\alpha=H+\frac12$ gives
\[
V_t =  \xi_0(t) -  \frac{1}{\Gamma(\alpha)}  \int_{0}^t (t-s)^{\alpha-1} \lambda \,\left(V_s - \xi_0(s)\right) ds +\frac{1}{\Gamma(\alpha)}  \int_{0}^t (t-s)^{\alpha-1} \nu \sqrt{V_s} dB_s.
\]
From Lemma \ref{lem:34proof}, the solution is given by
\begin{equation} \label{eq:variance2} 
V_t = \xi_0(t) + \int_0^t f^{\alpha, \lambda}(t-s) \frac{\nu}{\lambda} \sqrt{V_s} dB_s
\end{equation}  
where 
$f^{\alpha, \lambda}(x) = \lambda x^{\alpha - 1} \sum_{k \geq 0} \frac{(-\lambda x^\alpha)^k}{\Gamma(\alpha(k+1))}$ is the Mittag-Leffler density function. 
It follows that the future integrated variance satisfies
\beas
\sigma^2_{t + k\, \delta} &=&\int_{t+(k-1) \,\delta}^{t + k \,\delta} V_s\,ds\\
&=& \int_{t+(k-1) \,\delta}^{t + k \,\delta} \xi_0(s)\,ds +  \int_{0}^{t + k\, \delta} F^{\alpha, \lambda}(t+k \,\delta -s) \frac{\nu}{\lambda} \sqrt{V_s} dB_s\\
&-&  \int_{0}^{t + (k-1)\, \delta} F^{\alpha, \lambda}(t+(k-1)\, \delta - s) \frac{\nu}{\lambda} \sqrt{V_s} dB_s,
\eeas
where $F^{\alpha, \lambda}(x) = \int_0^x f^{\alpha, \lambda}(s) ds$. 
Consequently \eqref{inter} may be rewritten as
$$ 
Z_t(k) = 2 \,\frac{\rho \nu}{\lambda} \int_{t-\delta}^t \left(F^{\alpha, \lambda}(t + k \delta - s) - F^{\alpha, \lambda}(t + (k-1) \delta - s)\right)\, \E\left[V_s \int_{t- \delta}^s \sqrt{V_u} dW_u \right] ds. 
$$
Again using \eqref{eq:variance2}, we have that for $s \geq t - \delta$,
$$ 
\E\left[V_s \int_{t- \delta}^s \sqrt{V_u} dW_u \right] = \frac{\rho \nu}{\lambda} \int_{t- \delta}^s f^{\alpha, \lambda}(s-u) \,\xi_0(u) \,du.
$$
Thus,
\bea \label{eq:final_form} 
Z_t(k)= 2 \,\frac{(\rho\, \nu)^2}{\lambda^2} \int_0^{\delta}\, \left(F^{\alpha, \lambda}(s + k \delta) - F^{\alpha, \lambda}(s + (k-1) \delta )\right)\, \int_0^{\delta-s} \,f^{\alpha, \lambda}(u) \xi_0(t-s-u) \,du\, ds,\nonumber\\ 
\eea
which is positive if $\rho$ is different from zero. 
Using the fact that 
$$ 
F^{\alpha, \lambda}(x) \underset{x \rightarrow 0}{\sim} \frac{\lambda \,x^{\alpha}}{\Gamma(\alpha + 1)}, 
$$ together with the dominated convergence theorem,
we derive the following result.
\begin{theorem} \label{theo_approx} Assume that $\rho$ is nonzero and that the forward variance curve is continuous. Then $Z_t(k) > 0$ and as $\delta$ goes to zero,
$$ Z_t(k)  \underset{\delta \rightarrow 0}{\sim} 2 (\rho \,\nu)^2 \,\delta^{2 \alpha + 1}\, g_{\alpha}(k)\, \xi_0(t), \quad k \in \N_{>0}, \quad t >0, $$
with $ g_{\alpha}(k) = \frac{1}{\Gamma(\alpha + 1)^2} \int_0^1 \left( (k+s)^\alpha - (k+s-1)^\alpha \right) (1-s)^\alpha ds.$
\end{theorem}
We see that this measure of the Zumbach effect is indeed independent of $\lambda$. It is also independent of $t$ in the flat forward variance curve case.

\section{Numerical results}\label{sec:numerical}

To compare model computations with empirical estimates, we adopt the following model parameters typical of calibrations to the SPX implied volatility surface:
$$ 
 \quad \rho = - 0.7, \quad \nu = 0.45, \quad H = 0.05, \quad \lambda = 0.3. 
 $$
We assume a flat forward variance curve setting $\xi_0(t)=0.025$, the approximate sample mean of ${\sigma_t}^2$.

 \begin{figure} [tbh!]
\centering
\includegraphics[width= 0.8\linewidth ]{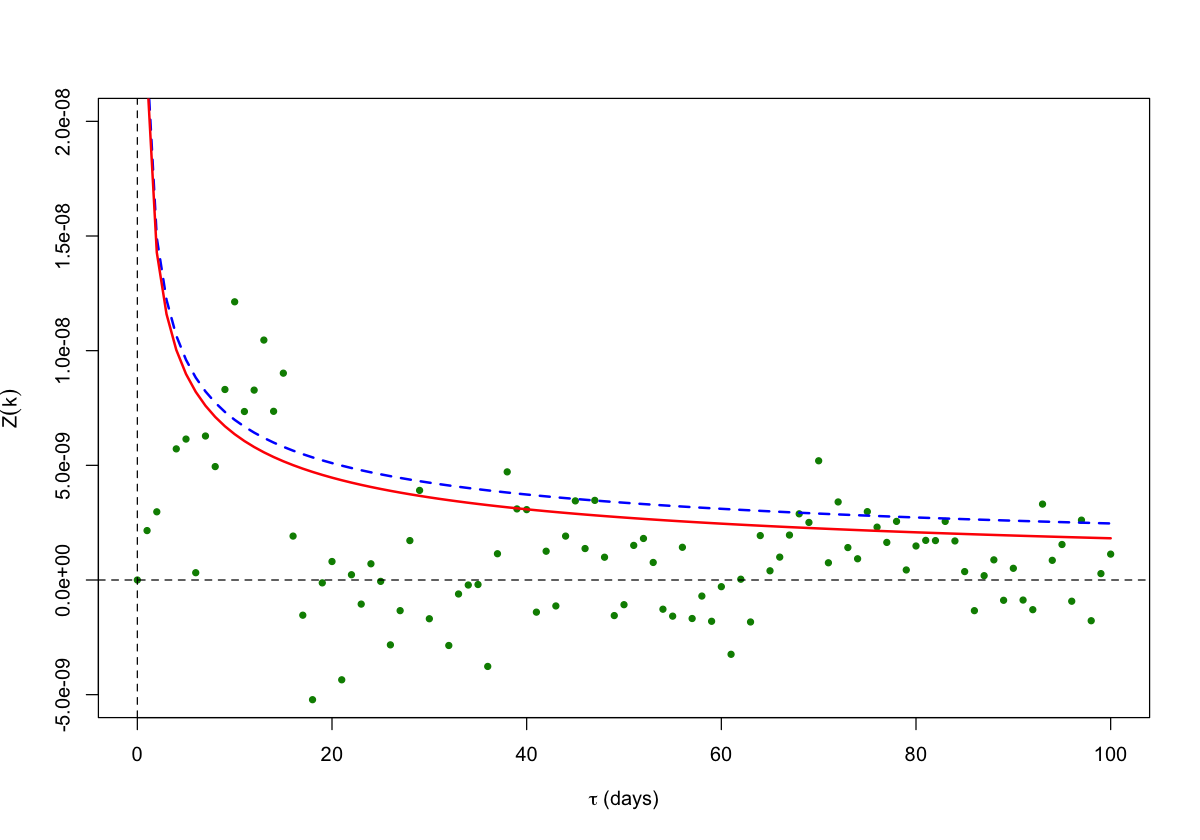}
\caption{With $\tau = k\,\delta$, the green points are the empirical estimates of $Z(\tau)$, the solid red line is the model computation \eqref{eq:final_form} of $Z_t(k)$, the dashed blue line is the small $\delta$ approximation from Theorem \ref{theo_approx} to $Z_t(k)$.}
\label{fig:ZumbachModel}
\end{figure}

 In Figure \ref{fig:ZumbachModel}, we superimpose empirical estimates $Z(\tau) = Z(k \delta)$ and model computations $Z_t(k)$ for SPX (which do not depend on $t$ here).  Although model computations are somewhat higher than empirical estimates, we argue that this nevertheless represents good agreement between model and data.  One factor no doubt contributing to the discrepancy is that we expect volatility of volatility and correlation under $\mQ$ to be more extreme than their equivalents under $\mP$.

\subsection{Dependence on $H$}

\begin{figure} [tbh!]
\centering
\includegraphics[width=0.8\linewidth ]{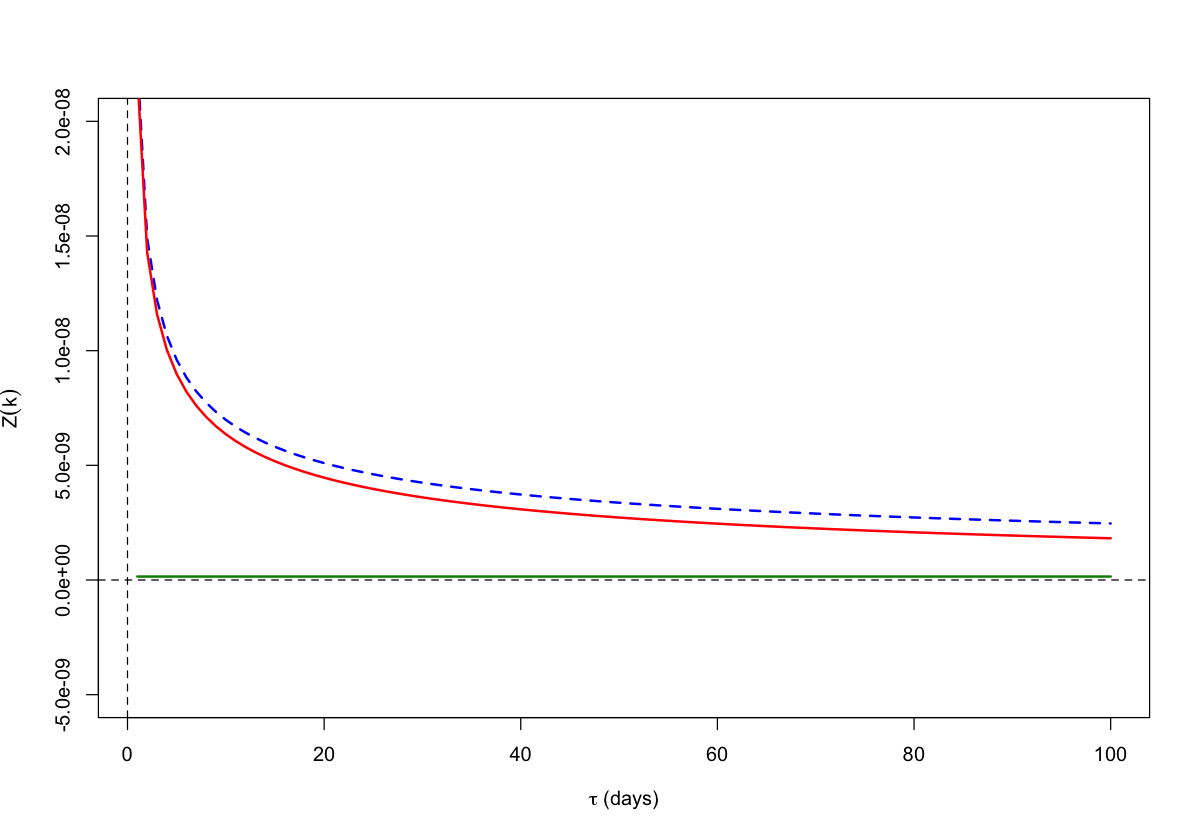}
\caption{The solid red line is $Z_t(k)$ computed with $H=0.05$, the blue dashed line is the approximation from Theorem \ref{theo_approx}, and the green line close to the $x$-axis is $Z_t(k)$ with $H=1/2$.  We see that the effect is negligible when $H=1/2$.}\label{fig:H05H5}
\end{figure}

We now examine the dependence of $Z_t(k)$ on the Hurst exponent $H$.  We already showed that under rough Heston with reasonable parameters, the Zumbach effect is consistent with empirical estimates.  In contrast, when $H=1/2$, we see from Figure \ref{fig:H05H5} that the Zumbach effect is negligible.  Indeed, from Theorem \ref{theo_approx}, $Z_t(k)$ is of order $\delta^{2\alpha+1} = \delta^{2H}$ for small $\delta$. When $H=1/2$, $Z_t(k) \sim \delta $ becomes very small, whereas as $H \rightarrow 0$, that is when volatility is rough, the Zumbach effect remains significant.

 \section*{Acknowledgements}
We thank Jean-Philippe Bouchaud for drawing the Zumbach effect to our attention, and for many animated and scintillating subsequent discussions.

\appendix

\section{List of indices in the Oxford-Man Institute of Quantitative Finance Realized Library}\label{sec:indices}

The following table lists all of the index tickers included in the Oxford-Man Institute of Quantitative Finance Realized Library together with index descriptions.

\begin{centering}
\begin{tabular}{ll}
  \hline
Index ticker & Index description\\ 
  \hline
.AEX & Amsterdam Exchange Index\\ 
  .AORD & All Ordinaries Index\\ 
  .BFX & BEL 20 Index \\ 
  .BSESN & S\&P Bombay Stock Exchange SENSEX Index\\ 
  .BVLG & Euronext PSI General Index \\ 
  .BVSP & BOVESPA Index \\ 
  .DJI & Dow Jones Industrial Average \\ 
  .FCHI & CAC 40 \\ 
  .FTMIB & FTSE MIB Index \\ 
  .FTSE & FTSE 100 Index \\ 
  .GDAXI & DAX  Index\\ 
  .GSPTSE & S\&P/TSX Composite Index \\ 
  .HSI & Hang Seng Index\\ 
  .IBEX & IBEX 35 Index\\ 
  .IXIC & Nasdaq Composite Index \\ 
  .KS11 & KOSDAQ Composite Index \\ 
  .KSE & Karachi Stock Exchange 100 Index \\ 
  .MXX & Mexican Bolsa IPC Index \\ 
  .N225 & Nikkei 225 Index\\ 
  .NSEI & NIFTY 50 Index\\ 
  .OMXC20 & OMX Copenhagen 20 Index  \\ 
  .OMXHPI & OMX Helsinki All-Share Index\\ 
  .OMXSPI & OMX Stockholm All-Share Index \\ 
  .OSEAX & Oslo B{\o}rs All-Share Index  \\ 
  .RUT & Russell 2000 Index \\ 
  .SMSI & Madrid General Index \\ 
  .SPX & 	S\&P 500 Index\\ 
  .SSEC & Shanghai Composite Index \\ 
  .SSMI & Swiss Market Index\\ 
  .STI & Straits Times Index \\ 
  .STOXX50E & Euro STOXX 50 Index \\ 
   \hline
\end{tabular}
\end{centering}

 \section{Proof of \eqref{eq:variance2}}
The following technical lemma slightly extends Proposition 4.10 in \cite{el2018microstructural}.
\begin{lemma}\label{lem:34proof}
The process $V$ is solution of the following rough stochastic differential
equation
\[
V_t =  \xi_0(t) -  \frac{1}{\Gamma(\alpha)}  \int_{0}^t (t-s)^{\alpha-1} \lambda \,\left(V_s - \xi_0(s)\right) ds +\frac{1}{\Gamma(\alpha)}  \int_{0}^t (t-s)^{\alpha-1} \nu \sqrt{V_s} dB_s
\]
if and only if it is solution of
\[
V_t = \xi_0(t) + \int_0^t f^{\alpha, \lambda}(t-s) \frac{\nu}{\lambda} \sqrt{V_s} dB_s.
\]
\end{lemma}

\begin{proof}
Suppose
\[
V_t = \xi_0(t) + \int_0^t f^{\alpha, \lambda}(t-s) \frac{\nu}{\lambda} \sqrt{V_s} dB_s.
\]
Then\footnote{For definitions and properties of fractional integration and differentiation and of the Mittag-Leffler density $f^{\alpha,\lambda}$, see for example Appendices A3 and A4 of \cite{el2018microstructural}.} using fractional integration of order $1-\alpha$ (denoted by $I^{1-\alpha}$), the properties of the Mittag-Leffler density and the stochastic Fubini theorem, this is equivalent to
\beas
I^{1-\alpha} V_t &=& I^{1-\alpha}\xi_0(t) + \frac{\nu}{\lambda}\,\int_0^t\,I^{1-\alpha} f^{\alpha, \lambda}(t-s)  \sqrt{V_s} dB_s\\
&=& I^{1-\alpha}\xi_0(t) + \frac{\nu}{\lambda}\,\int_0^t \,\lambda\,\left(1-F^{\alpha,\lambda}(t-s)\right)\,  \sqrt{V_s} dB_s\\
&=& I^{1-\alpha}\xi_0(t) + \nu\,\int_0^t \, \sqrt{V_s} dB_s
-\nu\,\int_0^t \,dB_s\,\int_s^{t}\,f^{\alpha,\lambda}(u-s)\, \sqrt{V_s} du\\
&=& I^{1-\alpha}\xi_0(t) + \nu\,\int_0^t \, \sqrt{V_s} dB_s
-\nu\,\int_0^t \,du\,\int_0^{u}\,f^{\alpha,\lambda}(u-s)\, \sqrt{V_s} \,dB_s\\
&=& I^{1-\alpha}\xi_0(t) + \nu\,\int_0^t \, \sqrt{V_s} dB_s
-\lambda\,\int_0^t \,\left(V_u - \xi_0(u)\right)\,du.
\eeas
Finally, applying fractional differentiation of order $1-\alpha$ together with the stochastic Fubini theorem we deduce the result.
\end{proof}

\section{The Zumbach effect in terms of correlations in the stationary regime}\label{sec:appcorrel}
We now discuss the Zumbach effect in terms of correlations in the stationary regime, that is when $t$ goes to infinity. In particular, we suppose that $\xi_0(t)$ satisfies
\begin{equation} \label{assum_stat}
\xi_0(t) \longrightarrow  \xi_0(\infty),
\end{equation} 
as $t$ goes infinity for some $\xi_0(\infty)>0$. 
 From Theorem \ref{theo_approx}, we have that for small $\delta$, $Z_t(k)$ is equivalent to
$$ 2 (\rho \nu)^2 \delta^{2 \alpha + 1} g_{\alpha}(k) \xi_0(\infty).$$

Moreover, from Appendix \ref{computation}, the limit of $ \E[r_t^4] $ as $t$ goes to infinity is 
 \begin{align*}
  & \xi_0(\infty) \frac{12\rho^2 \nu^2}{\lambda^2} \int_{0}^{\delta} F^{\alpha, \lambda}(s)F^{\alpha, \lambda}(\delta-u) du  + 3 \xi_0(\infty)^2 \delta^2   +  \frac{3\nu^2}{\lambda^2}  \xi_0(\infty) \int_0^{\delta}F^{\alpha, \lambda}(u)^2  du \\
  &+ \frac{6 \nu^2}{\lambda^2} \xi_0(\infty) \int_0^{\infty} \left(\int_0^{\delta}\left(F^{\alpha,\lambda}(s+u) - F^{\alpha,\lambda}(u) \right) f^{\alpha,\lambda}(s+u)  ds \right) du.
 \end{align*}
 This limit is equivalent for small $\delta$ to 
 $$  3 \xi_0(\infty)^2 \delta^2 + \frac{3 \nu^2}{\lambda^2} \xi_0(\infty) \delta ^2\int_0^{\infty}  f^{\alpha,\lambda}(s)^2  ds .$$
 
 \noindent In the same way, from Appendix \ref{computation}, the limit of $ \text{Var}[\sigma^2_t] $ as $t$ goes to infinity is
 $$  \frac{\nu^2}{\lambda^2}\xi_0(\infty) \int_0^{\infty} \left( F^{\alpha, \lambda}(s+\delta) - F^{\alpha, \lambda}(s) \right)^2  ds +\frac{\nu^2}{\lambda^2}  \xi_0(\infty) \int_0^{\delta} F^{\alpha, \lambda}(s)^2  ds, $$
 which is equivalent to
 $$ \frac{\nu^2}{\lambda^2}\xi_0(\infty) \delta^2 \int_0^{\infty} f^{\alpha, \lambda}(s)^2 ds.  $$

 \noindent Let us now define the correlation based Zumbach effect $Z^{\text{Correl}}(k)$ by
$$ Z^{\text{Correl}}(k)  =  \frac{\underset{t \rightarrow\infty}{\lim}  Z_t(k)}{\sqrt{\underset{t \rightarrow\infty}{\lim} \text{Var}[\sigma^2_t]\text{Var}[r_t]}}.$$
 From previous computations, we deduce the following proposition.
 \begin{prop} \label{prop:zumbach_prop} We have
 $$ Z^{\emph{Correl}}(k)   \underset{\delta \rightarrow 0}{\sim}\frac{2 (\rho \nu)^2 \sqrt{\xi_0(\infty)}}{\sqrt{\frac{\nu^2}{\lambda^2} \int_0^{\infty} f^{\alpha, \lambda}(s)^2 ds}\sqrt{2 \xi_0(\infty)+ \frac{3 \nu^2}{\lambda^2} \int_0^{\infty}  f^{\alpha,\lambda}(s)^2  ds}}\delta^{2 \alpha - 1} g_{\alpha}(k).$$
 \end{prop}

\section{Variance computations}\label{computation}
We compute in this section $ \text{Var}[\sigma_t^2]$ and $\text{Var}[r_t^2]$.  Using \eqref{eq:variance2} together with the stochastic Fubini theorem, we have that
\begin{equation} \label{expr_v}
\sigma^2_t = \int_{t- \delta}^t \xi_0(s) ds + \int_0^t \frac{\nu}{\lambda} F^{\alpha, \lambda}(t-s) \sqrt{V_s} dB_s -  \int_0^{t-\delta} \frac{\nu}{\lambda} F^{\alpha, \lambda}(t-\delta-s) \sqrt{V_s} dB_s.
\end{equation}
Hence,
$$ \text{Var}[\sigma^2_t] = \frac{\nu^2}{\lambda^2} \int_0^{t-\delta} \left( F^{\alpha, \lambda}(s+\delta) - F^{\alpha, \lambda}(s) \right)^2 \xi_0(t-\delta-s) ds +\frac{\nu^2}{\lambda^2}  \int_0^{\delta} F^{\alpha, \lambda}(s)^2 \xi_0(t-s) ds. $$
In order to compute $\text{Var}[r_t^2]$, we need to get $\E[r_t^4]$. Note that by It\^o's formula, 
$$\E[r_t^4] = 6 \int_{t-\delta}^t  \E\left[\left(\int_{t-\delta}^s \sqrt{V_u} dW_u\right)^2 V_s\right] ds.$$
 Using again It\^o's formula, we get that $ \E\left[\left(\int_{t-\delta}^s \sqrt{V_u} dW_u\right)^2 V_s\right] $ is equal to  
\begin{equation} \label{expr_aux}
2  \E\left[V_s \int_{t-\delta}^s \sqrt{V_u} \int_{t-\delta}^u \sqrt{V_w} dW_wdW_u\right] +  \E\left[V_s \int_{t-\delta}^s V_u du\right].
\end{equation}
From \eqref{eq:variance2}, the first term  in \eqref{expr_aux} is given by 
 $$ \frac{2\rho \nu}{\lambda} \int_{t-\delta}^s f^{\alpha, \lambda}(s-u) \E\left[ V_u \int_{t-\delta}^u \sqrt{V_w} dW_w\right] du,$$
 which is equal to 
 $$ \frac{2\rho^2 \nu^2}{\lambda^2} \int_{t-\delta}^s f^{\alpha, \lambda}(s-u) \left(\int_{t-\delta}^u f^{\alpha, \lambda}(u-w) \xi_0(w) dw \right) du.$$
Hence the first term in \eqref{expr_aux} is equal to
 $$ \frac{2\rho^2 \nu^2}{\lambda^2} \int_0^{s- t+\delta} f^{\alpha, \lambda}(u) \left(\int_0^{s-u-t+\delta} f^{\alpha, \lambda}(w) \xi_0(s-u-w) dw \right) du .$$ 
 Moreover, similarly to \eqref{expr_v}, 
 $$\int_{t-\delta}^s V_u du  = \int_{t- \delta}^s \xi_0(u) du + \int_0^s \frac{\nu}{\lambda} F^{\alpha, \lambda}(s-u) \sqrt{V_u} dB_u -  \int_0^{t-\delta} \frac{\nu}{\lambda} F^{\alpha, \lambda}(t-\delta-u) \sqrt{V_u} dB_u . $$ 
 Therefore the second term of \eqref{expr_aux} is given by
\begin{align*} &\xi_0(s)  \int_{0}^{s-t+\delta} \xi_0(u+t-\delta) du  + \frac{\nu^2}{\lambda^2} \int_0^{s-t+\delta} f^{\alpha, \lambda}(u) F^{\alpha, \lambda}(u) \xi_0(s-u) du\\
& + \frac{\nu^2}{\lambda^2} \int_0^{t-\delta} \left(F^{\alpha,\lambda}(s-t+\delta+u) - F^{\alpha,\lambda}(u) \right) f^{\alpha,\lambda}(s-t+\delta+u) \xi_0(t-\delta-u) du.
 \end{align*}
 Consequently, $\mathbb{E}[r_t^4]$ is equal to
 \begin{align*}
  &\frac{12\rho^2 \nu^2}{\lambda^2} \int_{0}^{\delta} \int_0^{s} f^{\alpha, \lambda}(u) \left(\int_0^{s-u} f^{\alpha, \lambda}(w) \xi_0(s-u+t-\delta-w) dw \right) du ds \\
  &+ 6 \int_0^{\delta} \xi_0(s+t-\delta)  \int_{0}^{s} \xi_0(u+t-\delta) du ds +  \frac{6\nu^2}{\lambda^2} \int_0^{\delta} \int_0^{s} f^{\alpha, \lambda}(u) F^{\alpha, \lambda}(u) \xi_0(s+t-\delta-u) duds \\
  &+ \frac{6 \nu^2}{\lambda^2} \int_0^{t-\delta} \left(\int_0^{\delta}\left(F^{\alpha,\lambda}(s+u) - F^{\alpha,\lambda}(u) \right) f^{\alpha,\lambda}(s+u)  ds \right) \xi_0(t-\delta-u) du.
 \end{align*}

\bibliographystyle{alpha}
\bibliography{Zumbach}

\begin{thebibliography}{BBMZ05}

\bibitem[BBMZ05]{borland2005dynamics}
Lisa Borland, Jean-Philippe Bouchaud, Jean-Fran{\c{c}}ois Muzy, and Gilles
  Zumbach.
\newblock The dynamics of financial markets-{M}andelbrot's cascades and beyond.
\newblock {\em Wilmott Magazine}, 2005.

\bibitem[BDB17]{blanc2017quadratic}
Pierre Blanc, Jonathan Donier, and Jean-Philippe Bouchaud.
\newblock Quadratic {H}awkes processes for financial prices.
\newblock {\em Quantitative Finance}, 17(2):171--188, 2017.

\bibitem[BLP16]{bennedsen2016decoupling}
Mikkel Bennedsen, Asger Lunde, and Mikko~S Pakkanen.
\newblock Decoupling the short-and long-term behavior of stochastic volatility.
\newblock {\em Available at SSRN 2846756}, 2016.

\bibitem[CB14]{chicheportiche2014fine}
R{\'e}my Chicheportiche and Jean-Philippe Bouchaud.
\newblock The fine-structure of volatility feedback {I}: Multi-scale
  self-reflexivity.
\newblock {\em Physica A: Statistical Mechanics and its Applications},
  410:174--195, 2014.

\bibitem[EEFR18]{el2018microstructural}
Omar El~Euch, Masaaki Fukasawa, and Mathieu Rosenbaum.
\newblock The microstructural foundations of leverage effect and rough
  volatility.
\newblock {\em Finance and Stochastics}, 22(2):241--280, 2018.

\bibitem[EEGR18]{el2018roughening}
Omar El~Euch, Jim Gatheral, and Mathieu Rosenbaum.
\newblock Roughening {H}eston.
\newblock {\em Available at SSRN 3116887}, 2018.

\bibitem[EER18a]{el2018characteristic}
Omar El~Euch and Mathieu Rosenbaum.
\newblock The characteristic function of rough {H}eston models.
\newblock {\em Mathematical Finance, forthcoming}, 2018.

\bibitem[EER18b]{euch2018perfect}
Omar El~Euch and Mathieu Rosenbaum.
\newblock Perfect hedging in rough {H}eston models.
\newblock {\em The Annals of Applied Probability, forthcoming}, 2018.

\bibitem[GJR18]{gatheral2018volatility}
Jim Gatheral, Thibault Jaisson, and Mathieu Rosenbaum.
\newblock Volatility is rough.
\newblock {\em Quantitative Finance}, 18(6):933--949, 2018.

\bibitem[GO10]{gatheral2010zero}
Jim Gatheral and Roel~CA Oomen.
\newblock Zero-intelligence realized variance estimation.
\newblock {\em Finance and Stochastics}, 14(2):249--283, 2010.

\bibitem[GR18]{gatheral2018rational}
Jim Gatheral and Rado{\v s} Radoi{\v c}i\'c.
\newblock {Rational approximation of the rough Heston solution}.
\newblock {\em SSRN}, 2018.

\bibitem[JE18]{jaber2018markovian}
Eduardo~Abi Jaber and Omar~El Euch.
\newblock Markovian structure of the {V}olterra {H}eston model.
\newblock {\em arXiv preprint arXiv:1803.00477}, 2018.

\bibitem[LZ03]{lynch2003market}
Paul~E Lynch and Gilles~O Zumbach.
\newblock Market heterogeneities and the causal structure of volatility.
\newblock {\em Quantitative Finance}, 3(4):320--331, 2003.

\bibitem[ZL01]{zumbach2001heterogeneous}
Gilles Zumbach and Paul Lynch.
\newblock Heterogeneous volatility cascade in financial markets.
\newblock {\em Physica A: Statistical Mechanics and its Applications},
  298(3-4):521--529, 2001.

\bibitem[Zum04]{zumbach2004volatility}
Gilles Zumbach.
\newblock Volatility processes and volatility forecast with long memory.
\newblock {\em Quantitative Finance}, 4(1):70--86, 2004.

\bibitem[Zum09]{zumbach2009time}
Gilles Zumbach.
\newblock Time reversal invariance in finance.
\newblock {\em Quantitative Finance}, 9(5):505--515, 2009.

\end{thebibliography}
\end{document}